\documentclass{scrartcl}

\usepackage[a4paper,margin=1in, left=7em]{geometry}
\usepackage{amsmath, amssymb}
\usepackage{mathrsfs}
\usepackage{enumitem}
\usepackage[bibliography=common]{apxproof}
\usepackage{stmaryrd}
\usepackage{graphicx}

\usepackage{mathtools}
\usepackage{hyperref}
\hypersetup{
    colorlinks,
    linkcolor={red!50!black},
    citecolor={blue!50!black},
    urlcolor={blue!30!black}
}

\usepackage[
    type={CC},
    modifier={by},
    version={4.0},
]{doclicense}

\usepackage[capitalise,nameinlink,noabbrev]{cleveref}

\usepackage{xcolor}

\newcommand\calL{\mathcal{L}}

\usepackage{amsthm}
\theoremstyle{plain}
\newtheorem{theorem}{Theorem}[section]

\theoremstyle{definition}
\newtheorem{definition}[theorem]{Definition}
\newtheorem{proposition}[theorem]{Property}

\AddToHook{env/defn/begin}{\crefalias{theorem}{defn}}
\AddToHook{env/proposition/begin}{\crefalias{theorem}{proposition}}
\AddToHook{env/coro/begin}{\crefalias{theorem}{coro}}
\AddToHook{env/lemma/begin}{\crefalias{theorem}{lemma}}
\AddToHook{env/claim/begin}{\crefalias{theorem}{claim}}
\AddToHook{env/exmp/begin}{\crefalias{theorem}{exmp}}
\crefname{subsection}{Subsection}{Subsections}

\title{Locality Testing for NFAs is PSPACE-complete}
\author{Antoine Amarilli$^1$, Mikaël Monet$^1$, Rémi De Pretto$^{1,2}$\\
$^1$: Univ.\ Lille, Inria, CNRS, Centrale Lille, UMR 9189 CRIStAL\\
$^2$: École supérieure de chimie, physique, électronique de Lyon}
\date{}

\begin{document}

\maketitle

\begin{abstract}
  \textbf{Abstract.} The class of local languages is
  a well-known subclass of the regular languages that
  admits many equivalent characterizations. In this
  short note we establish the PSPACE-completeness of
  the problem of determining, given as input a
  nondeterministic finite automaton (NFA) $A$,
  whether the language recognized by~$A$ is local or
  not. This contrasts with the case of deterministic
  finite automata (DFA), for which the problem is
  known to be in PTIME.
\end{abstract}

\section{Introduction}

A language $L\subseteq \Sigma^*$ over some alphabet
$\Sigma$ is \emph{local} if it can be defined exactly
by specifying which letters can start words, which
letters can end words, which factors of size two 
can appear in words, and whether the empty word belongs to~$L$. Equivalently, local languages
can be defined as those recognized by so-called
\emph{local
DFAs}~\cite{biblio:Sak09,pin2010mathematical} or by
\emph{read-once $\varepsilon$NFAs}~\cite{biblio:AGMM25}. They can also be
characterized as the languages that are
\emph{letter-Cartesian}~\cite{biblio:AGMM25,pin2010mathematical}.
Local languages have mostly been studied for their
logical and language-theoretic properties, but they also
have found applications in data management,
e.g., for validating XML documents~\cite{balmin2004incremental}, for studying the resilience problem for
graph query languages~\cite{biblio:AGMM25}, or for
probabilistic query
evaluation~\cite{amarilli2025approximating}.

In this note we are interested in the complexity of
determining whether a language is local. A priori, the
complexity of this problem might depend on
the formalism used to represent the input language. It is
known, for instance, that this task can be solved in
PTIME if the language is represented as a
deterministic finite automaton
(DFA): this is shown in~\cite[Theorem~4.1]{balmin2004incremental} or in~\cite[Proposition 3.5]{biblio:AGMM25}. The goal
of this note is to give a self-contained proof of the fact that, when the input language is given as a
nondeterministic finite automata (NFAs), then problem becomes PSPACE-complete.
Formally, we show:

\begin{theorem}
  \label{thm:intro}
  The following problem is PSPACE-complete: given as
  input an NFA $A$, decide if its language $\calL(A)$ is local.
\end{theorem}

We show that PSPACE-hardness already holds
when $\calL(A)$ is assumed to be \emph{infix-free},
i.e., when no word of the language is a strict infix
of another word of the language.
Such languages have been studied for instance in~\cite{han2006infix}: they are also called \emph{reduced}
in~\cite{biblio:AGMM25}.
We further remark
that the same result holds when the language is
represented as a regular expression. The membership
in PSPACE follows from known results, and we prove
PSPACE-hardness by reducing from the NFA universality
problem, using a technique inspired from Greibach's
theorem~\cite{biblio:Greib68}.

\paragraph*{Structure of the note.} We formally define
local languages in Section~\ref{sec:prelims}, where
we also recall some known properties about them and
explain why the locality problem for NFAs is in
PSPACE. We show the lower bound in
Section~\ref{sec:lower}.

\section{Preliminaries and Problem Statement}
\label{sec:prelims}

\paragraph*{Alphabets, words, languages.} An \emph{alphabet} $\Sigma$ is a
finite set whose elements are called \emph{letters}. A \emph{word}
over~$\Sigma$ is a finite sequence of letters.
We write 
$\varepsilon$
for the empty word.
The set of words over the alphabet~$\Sigma$ is
denoted $\Sigma^*$,
and a \emph{language} is a subset of $\Sigma^*$. 
For two words $w, w' \in \Sigma^*$ we write $w w'$ for their concatenation.

For $w \in \Sigma^*$, a word $u \in \Sigma^*$ is an \emph{infix} of $w$ if
there exist $v, v' \in \Sigma^*$ such that $vuv' = w$. It is a \emph{strict}
infix if furthermore we have $vv' \neq \varepsilon$. We say that a language $L$ is
\emph{infix-free} if for every word $w \in L$, there is no strict infix of $w$
that belongs to $L$. 

\paragraph*{Regular languages and automata.} We
assume that the reader is familiar with 
regular languages, regular expressions, and
(deterministic or nondeterministic) finite automata.
If $A$ is such an
automaton, we denote by $\calL(A)$ the language it
recognizes.
We use the standard notions of language concatenation, language union (denoted
$+$), and Kleene star.

\paragraph*{Locality.}
Local languages were introduced in
\cite{biblio:MP71}. A \emph{local} language $L$ over
an alphabet $\Sigma$ is a language defined by three
finite sets $R$, $S$, and $F$, where $R$ is the set of
the first letters, $S$ is the set of last letters and
$F$ is the set of infixes of length 2 that are not
allowed. Additionally, the empty word may belong to~$L$ or not.
So, a non-empty word belonging to the language~$L$
begins with a letter of $R$, finishes with a letter
of $S$ and contains no infix of length 2 in $F$.
Formally, $L \setminus \{\varepsilon\} = (R\Sigma^* \cap \Sigma^*S) \setminus
\Sigma^*F\Sigma^*$. Local languages have many other
characterizations, and in this note we will use the
following:

\begin{definition}
  \label{def:cart}
  A language $L$ is \emph{letter-Cartesian} if the
  following implication holds: for every letter $x
  \in \Sigma$ and words $\alpha, \beta, \gamma,
  \delta \in \Sigma^*$, if $\alpha x \beta \in L$ and
  $\gamma x \delta \in L$, then $\alpha x \delta \in
  L$. 
\end{definition}

Then we have:

\begin{proposition}
  \label{prp:equiv}
A language is local if, and only if, it is
letter-Cartesian.
\end{proposition}

This equivalence is given as an exercise in
\cite[Chp.\ II, Ex.\ 1.8, (d)]{biblio:Sak09}, and a
full proof can be found in~\cite[Proposition
B.7]{biblio:AGMM25}.

We now recall why it is PSPACE to test if the language recognized
by an NFA is local. The following is
again shown in \cite{biblio:AGMM25}:

\begin{proposition}[\mbox{\cite[Lemma B.4]{biblio:AGMM25}}]
Given as input an NFA $A$, we can compute in PTIME
another NFA $A'$ satisfying that $\calL(A)$ is
local if and only if $\calL(A') \subseteq \calL(A)$.
\end{proposition}

From this it follows that testing locality of the
language recognized by an NFA is in PSPACE, because
testing NFA inclusion is in PSPACE. Hence what remains to
show in this note is the lower bound for
Theorem~\ref{thm:intro}.

\section{Lower Bound}
\label{sec:lower}

We now prove the lower bound for
Theorem~\ref{thm:intro}, and show that it holds
already for infix-free languages. For this we use an
adaptation of the proof of Greibach's theorem.
Formally:

\begin{proposition}
  The following problem is PSPACE-hard: given as
  input an NFA $A$, determine if~$\calL(A)$ is local.
  This holds even when the input $A$ is required to ensure
  that~$\calL(A)$ is infix-free.
\end{proposition}
\begin{proof}
  We reduce from the PSPACE-hard problem of \emph{NFA
  universality}~\cite{biblio:MS72}. Recall that, in this problem, we are given as input an NFA $A$ over some
  alphabet $\Sigma$, and we must determine whether $A$ is
  universal, i.e., whether $\calL(A) = \Sigma^*$.
  Let us write $L=\calL(A)$. We can assume without
  loss of generality that $L\neq \emptyset$, because this test (the NFA
  emptiness problem) can be done in PTIME in $A$ simply by reachability
  testing.

  Let $a, \#_1, \#_2, \#_3$ be
  fresh letters not in $\Sigma$. From $A$, we build
  in PTIME another NFA $A'$ over the alphabet $\Sigma'\coloneq \Sigma \cup \{a, \#_1, \#_2, \#_3\}$ 
  for
  the language 
  $L' \coloneq \#_1(a^*\#_2 L + aa \#_2
  \Sigma^*)\#_3$.
  Observe that $L'$ is infix-free,
  thanks to the delimiters~$\#_1$ and~$\#_3$. We now
  prove that $L'$ is local if and only if $L =
  \Sigma^*$, i.e., that the reduction is correct. We
  prove both directions in turn.
  \begin{description}
    \item[“If” direction.] Assume that $L=\Sigma^*$.
      Then we have $L' = \#_1 a^* \#_2 \Sigma^*
      \#_3$, and it is easy to see (using
      Proposition~\ref{prp:equiv} for instance) that
      $L'$ is local.
    \item[“Only if” direction.] We prove the contrapositive: assume that $L \neq
      \Sigma^*$, and let us prove that $L'$ is not
      local. Since $L \neq \Sigma^*$, let $w$ be a word in $\Sigma^* \setminus
      L$. Further, since we assumed that $L \neq \emptyset$, let $u$ be a word
      in~$L$. Consider the two 
      words $w_1 \coloneq \#_1 \#_2 u \#_3$ and $w_2 \coloneq \#_1 aa \#_2 w
      \#_3$. The words $w_1$ and $w_2$ are both in~$L'$, but the word $w_3
      \coloneq \#_1
      \#_2 w \#_3$ is not in~$L'$. Therefore~$L'$ is
      not letter-Cartesian (using $x = \#_2$, $\alpha
      = \#_1$, $\beta = u \#_3$, $\gamma = \#_1 aa$,
      and $\delta = w \#_3$ in
      Definition~\ref{def:cart}), so that by
      Proposition~\ref{prp:equiv} $L'$ is not local.
      This concludes the proof.\qedhere
  \end{description}
\end{proof}

\bibliographystyle{plain}
\bibliography{main}

\vfill
\doclicenseThis

\end{document}